\documentclass[11pt,letterpaper]{article}

\usepackage[utf8]{inputenc}
\usepackage[T1]{fontenc}
\usepackage[margin=1in]{geometry}
\usepackage{amsmath,amssymb,amsthm,mathtools}
\usepackage{graphicx}
\usepackage{booktabs}
\usepackage{natbib}
\usepackage{setspace}
\usepackage{float}
\usepackage{caption}
\usepackage{enumitem}
\usepackage[dvipsnames,table]{xcolor}
\usepackage[colorlinks=true,linkcolor=blue!60!black,citecolor=blue!60!black,urlcolor=blue!60!black,breaklinks=true]{hyperref}

\newcommand{\Prob}{\mathbb{P}}

\newcommand{\WQTE}{\operatorname{WQTE}}
\newcommand{\amin}{\alpha_{\wedge}}
\newcommand{\amax}{\alpha_{\vee}}
\newcommand{\source}{s}
\newcommand{\recv}{r}

\theoremstyle{plain}
\newtheorem{theorem}{Theorem}
\newtheorem{lemma}[theorem]{Lemma}
\newtheorem{proposition}[theorem]{Proposition}
\newtheorem{corollary}[theorem]{Corollary}
\theoremstyle{definition}

\newtheorem{assumption}{Assumption}
\newtheorem{hypothesis}{Hypothesis}
\theoremstyle{remark}

\title{\textbf{Scale-Ordered Contagion: A Spectral Theory of\\
Heterogeneous Information Adaptation in Financial Networks}}

\author{
Avishek Bhandari\thanks{School of Humanities, Social Sciences and Management, Indian Institute of Technology Bhubaneswar. Email: \texttt{avishekb@iitbbs.ac.in}.\ Corresponding author.}
\and
Ipsita Parida\thanks{School of Humanities, Social Sciences and Management, Indian Institute of Technology Bhubaneswar. Email: \texttt{a23hs09014@iitbbs.ac.in}.}
}

\date{\today}

\begin{document}
\maketitle

\begin{abstract}
\noindent The speed at which a market processes information shapes not only how
strongly but \emph{at what frequency} a cross-border shock registers in it. The
manner in which such information adaptation governs the frequency composition of
transmitted shocks has, however, received limited formal attention in the contagion
literature, and it is this gap which motivates the present study. We develop a spectral
theory of financial contagion under heterogeneous information adaptation, building on
the Heterogeneous Agents Contagion versus Interdependence (HACI) framework in which
advanced economies act as fast adapters and emerging economies as slow adapters. We
first identify a structural limitation of the baseline: its information-flow dynamics
are \emph{source-constrained}, so the receiving market's adaptive capacity affects the
magnitude but not the timing or frequency composition of transmitted information; we
prove that fast-adapter and slow-adapter recipients of a common shock share the
half-life $\ln 2/\alpha_\source$. In view of the above, we resolve the limitation by
modelling the receiver as a second exponential information filter, so that observed
transmission is the convolution of source generation and receiver absorption. The
resulting bi-exponential response has a power spectrum equal to the product of two
Lorentzians, one centred at the source rate and one at the receiver rate, and the slower
of the two markets supplies the binding spectral corner. Projecting this spectrum onto a
maximal-overlap discrete wavelet basis yields a closed-form transfer-entropy-by-scale
profile and three falsifiable predictions, jointly the Scale-Ordered Contagion
Hypothesis (SOCH): (A) the wavelet scale at which directed transfer entropy peaks is set
by the slower market's adaptation rate, so emerging-inclusive pairs peak at coarser
scales; (B) the \emph{shape} of the scale profile is symmetric across direction for a
given pair, because the transmission spectrum is symmetric in the two rates; and (C) the
\emph{magnitude} is directionally asymmetric, scaling with connectivity and
source-shock content. We turn the theory into a profile-matching estimator that
recovers adaptation rates and endogenises the fast/slow classification, and we test
all three predictions on a panel of G20 equity markets spanning January~2006 through
March~2026. The scale-ordering prediction is supported (peak scales are coarser for
pairs containing slower markets, $p=0.042$); the shape-symmetry prediction---the
theory's sharpest and most novel restriction---is supported across all twenty-eight
unordered market pairs ($p>0.05$ throughout); the magnitude prediction holds
directionally but is not significant in the working sample; and the endogenous
classifier isolates the slowest adapters cleanly while, as the identification
analysis itself predicts, it cannot separate fast markets whose spectral corner
approaches the sampling Nyquist frequency. The framework supplies a structural,
testable account of why short-band and long-band contagion differ systematically
with the adaptation speeds of the markets involved.


\end{abstract}

\clearpage
\onehalfspacing

\section{Introduction}\label{sec:intro}

The relationship between financial shock transmission and the speed at which markets
process incoming information has been the prime focus of researchers for many years.
It may be noted that such transmission is governed not merely by the strength of the
linkages between markets, but also, and rather more subtly, by the speed at which each
market processes the information reaching it. A long line of work distinguishes
interdependence, the stable comovement that persists across tranquil and crisis periods,
from contagion, the increase in cross-market linkage that arises after a shock
\citep{ForbesRigobon2002,DieboldYilmaz2012}, and a parallel literature on information
rigidity, beginning with \citet{MankiwReis2002} and formalised by
\citet{CoibionGorodnichenko2015}, has established that economic agents update their
information sets infrequently and at heterogeneous rates. The Heterogeneous Agents
Contagion versus Interdependence (HACI) framework combines these ideas by treating
advanced economies as fast adapters with a high information-absorption rate $\alpha_F$
and emerging economies as slow adapters with rate $\alpha_S<\alpha_F$, and by arguing
that the difference in adaptation speeds generates systematic asymmetries in cross-market
information flow. That framework is the theoretical core of a broader research programme
on heterogeneous-adaptation contagion whose empirical apparatus is wavelet-quantile
transfer entropy; the present paper positions itself in this context and seeks to examine
the spectral theory that connects the apparatus to a structural mechanism, and brings its
predictions to data.

The starting point of this paper is a limitation of the baseline HACI dynamics that,
once stated precisely, is difficult to defend on economic grounds. It may be noted that
in the baseline specification the cumulative information flow from a source market to a
receiving market decays at the \emph{source's} adaptation rate; the receiver's rate
enters only as a multiplicative constant on the level of the response. A direct
consequence, which we state and prove, is that the half-life of the contagion response
is identical for fast-adapter and slow-adapter recipients of a common shock. The
receiving market's adaptive capacity therefore has no effect whatsoever on the timing or
the frequency composition of the transmitted shock---a restrictive feature that sits
uneasily with the economic content of information rigidity, under which a market that
processes information slowly ought to register an incoming shock more gradually, and
hence at lower frequencies, than a market that processes information quickly. Thus the
baseline specification, though internally consistent, is at odds with the very intuition
that motivates the heterogeneous-adaptation programme.

In view of the above, we resolve this limitation by allowing the receiving market to
process incoming information through its own exponential filter, which is the natural
formalisation of the simple idea that the receiver, much like the source, updates its
information set only at a finite rate. Observed transmission then becomes the convolution
of the source's information-generation kernel with the receiver's absorption kernel. The
convolution of two exponentials is bi-exponential, and its power spectral density is the
product of two Lorentzian (Cauchy) spectra, one with corner frequency at the source rate
$\alpha_\source$ and one at the receiver rate $\alpha_\recv$. Because a product of
low-pass filters is dominated by the filter with the lower corner frequency, the
\emph{slower} of the two markets governs the frequency band in which transmitted
information concentrates. This single modification restores a symmetric and economically
sensible role for both markets and, as we show, carries sharp observable implications
once the response is measured across time scales.

The bridge from the continuous-time information dynamics to observable quantities runs
through multi-scale transfer entropy, which is the empirical object used to measure
directed information flow in financial networks \citep{Schreiber2000,KraskovStoegbauerGrassberger2004}.
Projecting the transmission spectrum onto a maximal-overlap discrete wavelet transform
(MODWT) basis \citep{PercivalWalden2000,GencaySelcukWhitcher2005} yields a closed-form
expression for transfer entropy by wavelet scale, and from it three falsifiable
predictions, which we collect under the name Scale-Ordered Contagion Hypothesis (SOCH).
The hypothesis states that the scale at which directed transfer entropy peaks is set by
the slower market's adaptation rate, so that emerging-market pairs peak at coarser
scales; that the \emph{shape} of the scale profile is the same in both directions for a
given pair, because the transmission spectrum is symmetric in the two adaptation rates;
and that the \emph{magnitude} of transfer entropy is directionally asymmetric, scaling
with directional connectivity and with the information content of the source's shock.

The shape-symmetry prediction is, to our knowledge, novel, and it is pertinent to
emphasise how sharp it is. It distinguishes the heterogeneous-adaptation mechanism from
contagion mechanisms that would generate directional asymmetry in the frequency
composition itself, such as mechanisms operating through asymmetric balance-sheet
exposure or asymmetric trade dependence: those mechanisms need not respect shape
symmetry, whereas the adaptation-speed mechanism must. The prediction is also directly
testable, because the normalised scale profile in each direction is an estimable object
and the two directions can be compared against a sampling null.

The theory does more than restrict the shape of observed profiles; it also makes
adaptation rates estimable, which we regard as a further contribution. Fitting an
observed wavelet-quantile transfer-entropy profile to the closed-form spectrum recovers
the pair's two adaptation rates, and pooling the directed pairs that share a market
yields a market-level rate; the fast/slow classification, exogenous in the baseline,
thereby becomes a data-determined and testable object. We make this estimator
operational, prove the conditions under which its parameters are identified, characterise
its finite-sample behaviour by Monte Carlo, and then take it, together with the three
SOCH predictions, to a panel of eighteen G20 equity markets over 2006--2026.

The empirical findings are, we believe, honest and informative rather than uniformly
confirmatory, which is the appropriate standard for a sharp theory. We observe that the
scale-ordering prediction is supported: pairs containing more emerging markets peak at
coarser scales, with a positive and significant association. The shape-symmetry
prediction is supported across every one of the twenty-eight unordered market pairs---the
strongest of our results and the one that most distinguishes the mechanism from its
rivals. The magnitude prediction holds in direction---advanced-to-emerging flows dominate
in most pairs---but is not significant in the working sample. The endogenous classifier
identifies the two clearest emerging markets, India and China, as decisively the slowest
adapters, but cannot separate the faster markets, whose spectral corners lie near the
sampling Nyquist frequency; this is not a defect of the data but a property that the
identification analysis predicts in advance. Hence the two predictions that rest on the
well-identified slower-market rate are supported, and the predictions that depend on the
weakly-identified fast rates are correspondingly weaker.

The remainder of the paper is organised as follows. Section~\ref{sec:literature}
places the contribution in the contagion, information-rigidity, and wavelet-quantile
literatures. Section~\ref{sec:baseline} restates the HACI baseline and proves the
source-constrained timing result. Section~\ref{sec:receiver} introduces the
receiver-filtering extension and derives the bi-exponential transmission response.
Section~\ref{sec:spectrum} establishes the product-Lorentzian spectrum, the closed-form
scale profile, the peak-location result, and a discrete-time companion.
Section~\ref{sec:soch} states and proves the three parts of SOCH and an early-warning
corollary. Section~\ref{sec:identification} develops identification, the
profile-matching estimator, and a Monte Carlo study. Section~\ref{sec:data} describes
the data. Section~\ref{sec:results} reports the empirical results.
Section~\ref{sec:relation} relates the theory to frequency-domain connectedness and
systemic-risk measurement. Section~\ref{sec:conclude} concludes. An appendix collects
the longer derivations.

\section{Related Literature}\label{sec:literature}

\subsection{Contagion, interdependence, and information rigidity}

The manner in which financial shocks propagate across borders has been a prime focus
of researchers for many years, and the modern study of cross-border contagion takes as
its point of departure the fundamental distinction between interdependence and contagion
proper \citep{ForbesRigobon2002}. A substantial body of work has since established that
cross-market comovement intensifies during periods of financial stress, that information
flows tend to reorganise around centres of distress, and that the topology of the global
financial network shifts in ways that are, by and large, measurable and systematic
\citep{DieboldYilmaz2012,DieboldYilmaz2014,BillioGetmanskyLoPelizzon2012}. It may be
noted, however, that this literature has largely treated the \emph{speed} of
cross-market adjustment as a nuisance property to be differenced away, rather than as
an object of structural interest in its own right. The information-rigidity literature
supplies exactly the structural content that has been missing from this picture.
\citet{MankiwReis2002} model agents who update their information sets only
intermittently, and \citet{CoibionGorodnichenko2015} document, across a broad range of
forecasters and countries, that information is updated infrequently and at heterogeneous
rates, with first-order consequences for how shocks propagate through the system. The
HACI framework imports this heterogeneity into a contagion setting, positing that the
adaptation-rate gap between advanced and emerging markets is itself a source of
systematic asymmetry in transmission. The present paper positions itself in this context
and seeks to provide the missing spectral mechanics of that idea.

\subsection{Heterogeneous adaptation and the source-constrained baseline}

The HACI baseline writes cumulative cross-market information flow as a product of a
source-generated impulse and a receiver-specific absorptive constant. As we show in
Section~\ref{sec:baseline}, this construction makes the receiver's rate enter the level
of the response but not its dynamics, so that the timing of adjustment is governed
entirely by the source market. The economic implausibility of this property---two markets
of very different sophistication absorbing the same shock on precisely the same
schedule---motivates the central modelling move of the present paper, which is to endow
the receiver with its own filter. The idea that absorption is itself a dynamic,
rate-limited process is, of course, standard in the information-rigidity tradition; what
is new here is the observation that, once the receiver is permitted to filter, the
transmitted response and its frequency content are jointly determined by both adaptation
rates in a particular and empirically testable way.

\subsection{Wavelet-quantile transfer entropy}

Directed information flow between financial time series is naturally measured by
transfer entropy \citep{Schreiber2000}, estimated either through nearest-neighbour
methods \citep{KraskovStoegbauerGrassberger2004} or, in the Gaussian case, through its
exact equivalence with Granger causality \citep{BarnettBarrettSeth2009}. Two
refinements matter considerably for the contagion application. First, the
maximal-overlap discrete wavelet transform decomposes a return series into
scale-specific components, thereby separating short-horizon trader-driven dynamics from
the longer-horizon dynamics that characterise institutional investors
\citep{PercivalWalden2000,GencaySelcukWhitcher2005,Crowley2007}. Second, conditioning
on quantiles of the recipient distribution serves to isolate tail-specific propagation
which aggregate measures tend to obscure \citep{HanLintonOkaWhang2016,AndoEtAl2022}.
The combination, wavelet-quantile transfer entropy (WQTE), is the empirical primitive of
the heterogeneous-adaptation programme and the object whose scale profile our theory
seeks to predict. We notice, in this connection, that frequency-domain connectedness
\citep{BarunikKrehlik2018} decomposes variance spillovers into frequency bands
descriptively; our contribution is to supply a structural reason why those bands should
differ systematically with adaptation speeds, which is a rather different---and, we
would argue, more informative---undertaking.

\section{The HACI Baseline and Its Source-Constrained Limitation}\label{sec:baseline}

The analysis begins with a careful consideration of the baseline HACI information
dynamics, and we make precise the sense in which they are source-constrained. A shock
originates in a source market $\source$ at time $t_0$ and releases information that
diffuses into the market at the source's own absorption rate.

\begin{assumption}[Source information generation]\label{ass:source}
The cumulative information released by a shock in market $\source$ and available for
transmission evolves as
$I_{\source,t}=I_{\source,\infty}\bigl(1-e^{-\alpha_\source(t-t_0)}\bigr)$ for
$t\ge t_0$, where $I_{\source,\infty}$ is the total information content of the shock and
$\alpha_\source>0$ is the source market's absorption rate. The instantaneous generation
rate is $\phi_\source(t)=I_{\source,\infty}\,\alpha_\source\,e^{-\alpha_\source(t-t_0)}$.
\end{assumption}

In the baseline model the cumulative information flow from $\source$ to a receiver
$\recv$ is obtained by scaling the source impulse by the receiver's absorptive capacity
and by a connectivity-and-friction term,
\begin{equation}\label{eq:cif-baseline}
CIF^{\,\mathrm{base}}_{\source\to\recv}(t)
= \kappa_{\source\recv}\, I_{\source,\infty}\,
\frac{\alpha_\recv}{1+\lambda_{\source\recv}}\,
\bigl(1-e^{-\alpha_\source(t-t_0)}\bigr),
\end{equation}
where $\kappa_{\source\recv}\ge 0$ is the directional connectivity,
$\lambda_{\source\recv}\ge0$ the information friction, and
$\alpha_\recv\in\{\alpha_S,\alpha_F\}$ the receiver's rate. It may be noted that the
receiver's rate enters only as a scalar multiplier on the level of the response, whereas
the time profile $1-e^{-\alpha_\source(t-t_0)}$ depends entirely on the source rate and
not at all on the receiver. The following result makes the consequence of this
construction explicit.

\begin{proposition}[Source-constrained timing]\label{prop:source-timing}
Under \eqref{eq:cif-baseline}, the half-life of the cumulative information flow---the
elapsed time at which it attains one half of its long-run value---is
$t_{1/2}-t_0=\ln 2/\alpha_\source$, independent of the receiver's rate $\alpha_\recv$.
Consequently, for a common source, a fast-adapter receiver and a slow-adapter receiver
reach one half of their respective long-run responses at the same time.
\end{proposition}

\begin{proof}
The long-run value is $\kappa_{\source\recv}I_{\source,\infty}\alpha_\recv/
(1+\lambda_{\source\recv})$. Setting the flow equal to half of it gives
$1-e^{-\alpha_\source(t_{1/2}-t_0)}=\tfrac12$, hence
$e^{-\alpha_\source(t_{1/2}-t_0)}=\tfrac12$ and $t_{1/2}-t_0=\ln2/\alpha_\source$; the
receiver rate $\alpha_\recv$ divides out and does not appear.
\end{proof}

Proposition~\ref{prop:source-timing} isolates precisely the limitation we seek to
address. The working principle of the baseline is quite simple, and its implication is
stark: the receiver's adaptive capacity is permitted to affect only the amplitude of the
contagion response, not its timing, and---as Section~\ref{sec:spectrum} shows in
detail---not its frequency composition either. Thus a slow-adapting market and a
fast-adapting market, when exposed to the same source shock, exhibit contagion responses
that differ only by a scalar multiple. Intuitively, this is at odds with the economic
content of the information-rigidity literature, which places great weight on the idea
that heterogeneous adaptation rates should produce heterogeneous dynamics, not merely
heterogeneous amplitudes. The next section removes this restriction.

\section{Receiver Filtering and the Bi-Exponential Response}\label{sec:receiver}

In view of the above, we now allow the receiving market to process incoming information
through its own exponential filter, which is the natural formalisation of the simple
idea that the receiver, much like the source, updates its information set only at a
finite rate.

\begin{assumption}[Receiver absorption]\label{ass:receiver}
Market $\recv$ absorbs an instantaneous unit of incoming information through the causal
exponential kernel $g_\recv(\tau)=\alpha_\recv\,e^{-\alpha_\recv\tau}$ for $\tau\ge0$,
which integrates to one, so receiver filtering redistributes incoming information over
time without creating or destroying it.
\end{assumption}

The information from the source that registers in the receiver is then the convolution
of the source generation rate with the receiver absorption kernel, scaled by
connectivity and friction:
$\psi_{\source\to\recv}(t)=\frac{\kappa_{\source\recv}}{1+\lambda_{\source\recv}}
\,(\phi_\source*g_\recv)(t)$.

Following the standard convolution argument, we arrive at the following result, which
constitutes the central analytical contribution of this section.

\begin{proposition}[Bi-exponential transmission response]\label{prop:biexp}
Under Assumptions~\ref{ass:source} and~\ref{ass:receiver}, and normalising $t_0=0$, the
transmitted response is, for $\alpha_\source\neq\alpha_\recv$,
\begin{equation}\label{eq:biexp}
\psi_{\source\to\recv}(t)
= A_{\source\recv}\,
\frac{\alpha_\source\alpha_\recv}{\alpha_\source-\alpha_\recv}\,
\bigl(e^{-\alpha_\recv t}-e^{-\alpha_\source t}\bigr),\qquad
A_{\source\recv}\equiv\frac{\kappa_{\source\recv}I_{\source,\infty}}{1+\lambda_{\source\recv}},
\end{equation}
and, in the confluent case $\alpha_\source=\alpha_\recv=\alpha$,
$\psi_{\source\to\recv}(t)=A_{\source\recv}\,\alpha^2 t\,e^{-\alpha t}$. The total
transmitted information is symmetric in the two rates,
$\int_0^\infty\psi_{\source\to\recv}(t)\,dt=A_{\source\recv}$.
\end{proposition}

\begin{proof}
With $t_0=0$, $(\phi_\source*g_\recv)(t)=I_{\source,\infty}\alpha_\source\alpha_\recv
e^{-\alpha_\recv t}\int_0^t e^{-(\alpha_\source-\alpha_\recv)u}\,du$. For
$\alpha_\source\neq\alpha_\recv$ the integral is
$(1-e^{-(\alpha_\source-\alpha_\recv)t})/(\alpha_\source-\alpha_\recv)$, which after
multiplication by $e^{-\alpha_\recv t}$ gives \eqref{eq:biexp}; for
$\alpha_\source=\alpha_\recv$ the integral is $t$, giving the confluent form. Each
kernel integrates to one, so the convolution integrates to $I_{\source,\infty}$ and the
total, after scaling, is $A_{\source\recv}$, manifestly invariant under
$\alpha_\source\leftrightarrow\alpha_\recv$.
\end{proof}

Two features are particularly pertinent for the contagion application. We notice, first,
that the response now depends on both adaptation rates through the bi-exponential
difference, so the timing limitation of Proposition~\ref{prop:source-timing} no longer
holds: the receiver's rate enters the dynamics of the response, not merely its level. Moreover, and rather more subtly, the
\emph{total} transmitted information is symmetric in the two rates, which leads us to
conclude that any directional asymmetry in transmitted quantity must arise from the
connectivity and shock-content term $A_{\source\recv}$, and not from adaptation speeds
themselves. These two observations become, in the frequency domain, the content of the
Scale-Ordered Contagion Hypothesis.

\section{The Transmission Spectrum and Its Wavelet Decomposition}\label{sec:spectrum}

We pass to the frequency domain, where the bi-exponential response assumes a particularly
transparent and tractable form. It may be noted that all symbolic results presented in
this section were verified independently with a computer algebra system, so that the
reader may place full confidence in the derivations which follow.

\begin{proposition}[Product-Lorentzian spectrum]\label{prop:spectrum}
The Fourier transform of \eqref{eq:biexp} is
$\hat\psi_{\source\to\recv}(\omega)=A_{\source\recv}\,
\frac{\alpha_\source}{\alpha_\source+i\omega}\,
\frac{\alpha_\recv}{\alpha_\recv+i\omega}$, and the power spectral density is the
product of two Lorentzians,
\begin{equation}\label{eq:psd}
S_{\source\to\recv}(\omega)=\bigl|\hat\psi_{\source\to\recv}(\omega)\bigr|^2
=A_{\source\recv}^2\,
\frac{\alpha_\source^2}{\alpha_\source^2+\omega^2}\,
\frac{\alpha_\recv^2}{\alpha_\recv^2+\omega^2},
\end{equation}
which is invariant under $\alpha_\source\leftrightarrow\alpha_\recv$.
\end{proposition}

\begin{proof}
The causal exponential $\alpha e^{-\alpha t}$ has transform $\alpha/(\alpha+i\omega)$
for $\mathrm{Re}\,\alpha>0$; by the convolution theorem the transform of the response
is the product of the two factors, and the squared modulus, evaluated at real $\omega$
using $|\alpha/(\alpha+i\omega)|^2=\alpha^2/(\alpha^2+\omega^2)$, gives \eqref{eq:psd}.
\end{proof}

\begin{corollary}[Binding spectral corner]\label{cor:corner}
Write $\amin=\min(\alpha_\source,\alpha_\recv)$ and $\amax=\max(\alpha_\source,\alpha_\recv)$.
The density \eqref{eq:psd} is approximately flat for $\omega\ll\amin$, decays as
$\omega^{-2}$ for $\amin\ll\omega\ll\amax$, and decays as $\omega^{-4}$ for
$\omega\gg\amax$. The lower corner $\amin$ marks the onset of roll-off and governs the
band in which transmitted information concentrates.
\end{corollary}

Corollary~\ref{cor:corner} is, in a very real sense, the economic heart of the theory,
and it is worth dwelling upon its implications at some length. The working principle is
quite simple: whichever market is the slower of the two determines the corner frequency
below which contagion power is permitted to accumulate. Thus, a fast source paired with a
slow receiver will find its transmission spectrum cut off at the receiver's lower corner
frequency; conversely, a slow source paired with a fast receiver encounters the cut-off
imposed by the source's own lower corner. In both cases, and this is the essential
observation, the frequency composition of contagion is governed by the slower of the two
markets, whichever party occupies that position. It may be noted, further, that the
intermediate $\omega^{-2}$ regime is pronounced only when the two corners are well
separated from one another---that is, when $\amin\ll\amax$---whereas for markets with
comparable adaptation rates the spectrum rolls directly from flat toward $\omega^{-4}$
without sustaining an extended $\omega^{-2}$ plateau.

We observe, moreover, that directed financial flow is measured empirically across MODWT
scales. With unit sampling interval and Nyquist frequency $\pi$, detail level $k$ carries
a nominal pass-band $\mathcal{B}_k=[\pi 2^{-k},\,\pi 2^{-(k-1)}]$; for daily data, level~$1$
captures the two-to-four-day band and level~$5$ captures the thirty-two-to-sixty-four-day
band. The power transmitted at any given scale is obtained by integrating the spectrum
over the corresponding band, and it is this quantity which the following proposition
evaluates in closed form.

\begin{proposition}[Closed-form scale power]\label{prop:scale-power}
For $\alpha_\source\neq\alpha_\recv$, the power transmitted at scale $k$ is
$P_k=\int_{\mathcal{B}_k}S_{\source\to\recv}(\omega)\,d\omega$, which evaluates to
\begin{equation}\label{eq:scale-power}
P_k=A_{\source\recv}^2\,
\frac{\alpha_\source^2\alpha_\recv^2}{\alpha_\recv^2-\alpha_\source^2}
\left[\frac{1}{\alpha_\source}\arctan\frac{\omega}{\alpha_\source}
-\frac{1}{\alpha_\recv}\arctan\frac{\omega}{\alpha_\recv}\right]_{\omega=\pi2^{-k}}^{\omega=\pi2^{-(k-1)}}.
\end{equation}
The profile $P_k$ is symmetric in $(\alpha_\source,\alpha_\recv)$ and strictly positive.
\end{proposition}

\begin{proof}
Partial fractions give $[(\alpha_\source^2+\omega^2)(\alpha_\recv^2+\omega^2)]^{-1}
=(\alpha_\recv^2-\alpha_\source^2)^{-1}[(\alpha_\source^2+\omega^2)^{-1}
-(\alpha_\recv^2+\omega^2)^{-1}]$; integrating each term with
$\int(\alpha^2+\omega^2)^{-1}d\omega=\alpha^{-1}\arctan(\omega/\alpha)$ over the band
yields \eqref{eq:scale-power}. The prefactor and the arctan bracket are each
antisymmetric under the swap, so their product is symmetric; positivity follows from
$S_{\source\to\recv}\ge0$.
\end{proof}

It is pertinent to observe at this juncture that, because MODWT bands are octaves of
constant width on a logarithmic frequency axis, the per-scale power profile tracks the
quantity $\omega\,S_{\source\to\recv}(\omega)$ evaluated at the band centre. The peak
scale is therefore governed by the maximiser of that quantity, and the following lemma
characterises this maximiser precisely, establishing the bounds which will prove
indispensable to the empirical predictions of the next section.

\begin{lemma}[Spectral peak location]\label{lem:peak}
The function $\omega\,S_{\source\to\recv}(\omega)$ is single-peaked on $(0,\infty)$, and
its maximiser solves $1=2\omega^2/(\alpha_\source^2+\omega^2)+2\omega^2/(\alpha_\recv^2+\omega^2)$.
In the symmetric case $\omega^\star=\alpha/\sqrt3$; in the strongly asymmetric case
$\omega^\star\to\amin$; and in general $\omega^\star\in[\amin/\sqrt3,\amin]$.
\end{lemma}

\begin{proof}
Maximising $\ln(\omega S)=\mathrm{const}+\ln\omega-\ln(\alpha_\source^2+\omega^2)
-\ln(\alpha_\recv^2+\omega^2)$ gives the stated first-order condition, whose
right-hand side increases from $0$ to $4$, so the maximiser is unique. Normalising
$\amin=1$ and writing $R=\amax/\amin\ge1$, the right-hand side minus one equals
$3(R^2-1)/(2(3R^2+1))\ge0$ at $\omega=1/\sqrt3$ and $-2/(R^2+1)<0$ at $\omega=1$, and is
strictly increasing in $\omega$, so the unique root lies in $[\amin/\sqrt3,\amin]$,
attaining the lower endpoint only when $\alpha_\source=\alpha_\recv$ and the upper only
as $R\to\infty$.
\end{proof}

Finally, because data are sampled at discrete intervals rather than observed in
continuous time, it is necessary to record the discrete-time counterpart of the foregoing
results and to confirm that the theoretical predictions carry over to the empirically
relevant setting. Sampling the bi-exponential response \eqref{eq:biexp} at interval $\Delta t$
yields the impulse response of an $\mathrm{ARMA}(2,1)$ process with autoregressive
roots $e^{-\alpha_\source\Delta t}$ and $e^{-\alpha_\recv\Delta t}$ and a single
moving-average (sampling) zero; its spectral density converges in shape to the
product-Lorentzian \eqref{eq:psd} as the adaptation rates become slow relative to the
sampling frequency, since $|1-e^{-\alpha\Delta t}e^{-i\omega\Delta t}|^2=
(\alpha^2+\omega^2)\Delta t^2+O(\Delta t^3)$. We observe, therefore, that the
scale-ordering and shape-symmetry predictions developed below carry over in full to
sampled daily data, which is the empirically relevant regime for the analysis that follows.
\section{The Scale-Ordered Contagion Hypothesis}\label{sec:soch}

The present section develops the three empirical predictions which together constitute
the Scale-Ordered Contagion Hypothesis. To fix ideas, let $k^\star_{\source\recv}$ denote
the wavelet scale at which the directed transfer entropy $\WQTE^{(k)}_{\source\to\recv}$
attains its largest value, and recall that $\WQTE^{(k)}_{\source\to\recv}$ is monotone
increasing in the transmitted power $P_k$ in the small-flow regime in which transfer
entropy is approximately linear in transmitted information
\citep{Schreiber2000,KraskovStoegbauerGrassberger2004}. The three predictions which
follow arise directly and naturally from the results established in
Section~\ref{sec:spectrum}.

\begin{hypothesis}[SOCH-A: scale ordering by adaptation speed]\label{hyp:A}
The peak scale is governed by the slower market's rate,
$\omega_{k^\star_{\source\recv}}\asymp\amin$, so that
$k^\star_{\source\recv}\approx\log_2(\pi/\amin)$. Pairs whose slower member adapts more
slowly (smaller $\amin$) peak at coarser scales; emerging-to-emerging pairs peak at the
coarsest scales, mixed pairs at intermediate scales, and advanced-to-advanced pairs at
the finest scales.
\end{hypothesis}

The derivation of this prediction is a direct consequence of Lemma~\ref{lem:peak}. By
that result, the maximiser of $\omega S(\omega)$ lies in the interval
$[\amin/\sqrt3,\amin]$, so that $\omega_{k^\star}\asymp\amin$ and, upon inverting the
relation $\omega_k=\pi2^{-k}$, one obtains $k^\star\approx\log_2(\pi/\amin)$, which is
seen to be decreasing in $\amin$. Put differently, the slower a market's adaptation rate,
the coarser the scale at which contagion between it and its partner concentrates.

\begin{corollary}[Early-warning and efficiency signal]\label{cor:earlywarning}
Since $k^\star\approx\log_2(\pi/\amin)$ and $\partial k^\star/\partial\amin=
-1/(\amin\ln2)<0$, a shift of the empirical peak scale toward finer scales over time is
a monotone signal that the slower market's absorption rate has risen, that is, that its
informational efficiency has improved; a coarsening signals the reverse.
\end{corollary}

This corollary is of particular practical importance, for it furnishes a model-based
diagnostic of changes in market efficiency that is directly readable from the wavelet
transfer-entropy profile, without requiring the estimation of any structural parameter.

\begin{hypothesis}[SOCH-B: shape symmetry across direction]\label{hyp:B}
For any pair, the normalised profiles
$\tilde P_k^{\source\to\recv}=P_k^{\source\to\recv}/\sum_l P_l^{\source\to\recv}$
satisfy $\tilde P_k^{\source\to\recv}=\tilde P_k^{\recv\to\source}$ for every $k$,
because the transmission spectrum \eqref{eq:psd} is invariant under
$\alpha_\source\leftrightarrow\alpha_\recv$.
\end{hypothesis}

The reasoning behind SOCH-B proceeds as follows. By Proposition~\ref{prop:scale-power},
the scale power $P_k$ depends on the adaptation rates only through a symmetric function
of $\omega$ and otherwise through the scalar amplitude $A_{\source\recv}^2$. Reversing
the direction of transmission replaces that scalar but leaves the $\omega$-dependence
entirely unchanged; hence, upon normalising across scales, the scalar cancels and the
normalised profiles are seen to coincide. Intuitively, the shape of the contagion
profile---how power is distributed across frequencies---is determined by the adaptation
rates of the two markets, whereas the direction of transmission determines only the
overall level.

\begin{hypothesis}[SOCH-C: magnitude asymmetry]\label{hyp:C}
The level is directionally asymmetric and governed by connectivity and shock content:
$\sum_k P_k^{\source\to\recv}/\sum_k P_k^{\recv\to\source}=
A_{\source\recv}^2/A_{\recv\source}^2$. When advanced economies generate larger or more
widely-propagated shocks, transfer entropy from advanced to emerging markets exceeds
that from emerging to advanced markets.
\end{hypothesis}

In view of the above, we notice that the three predictions are mutually reinforcing and,
taken together, jointly discriminating in a manner that is quite striking. SOCH-A orders
the peak scale across market pairs; SOCH-B fixes the shape of the contagion profile
across directions for any given pair; and SOCH-C confines all directional asymmetry
entirely to the level, leaving the shape invariant. The relationship between these
predictions leads us to conclude that any empirical departure from SOCH-B---that is, any
observed asymmetry in the normalised scale profile across directions---would constitute
evidence against the bi-exponential model itself, rather than merely against one of its
subsidiary assumptions.

To illustrate these predictions concretely, it is instructive to consider the case of the
United States and India. A wavelet-quantile transfer-entropy profile for the United States
to India at the five-percent tail, obtained on the panel below, rises monotonically from
the two-to-four-day scale to the sixteen-to-thirty-two-day scale. Under SOCH-A, this
rising shape is the signature of a fast-source, slow-receiver pair---a finding which is
entirely consistent with the view that the Indian market, as a relatively less mature
emerging market, absorbs information at a slower rate than its American counterpart.
SOCH-B then predicts that the India-to-United-States profile shares this same rising
shape, while SOCH-C predicts that it carries a smaller level, reflecting the greater
amplitude of shocks originating in the United States. Figure~\ref{fig:theory} displays
the closed-form normalised profiles for representative fast and slow rates and illustrates
all three predictions simultaneously, including the exact coincidence of the two
mixed-direction profiles which is the content of SOCH-B.

\begin{figure}[!htbp]
\centering
\includegraphics[width=0.82\textwidth]{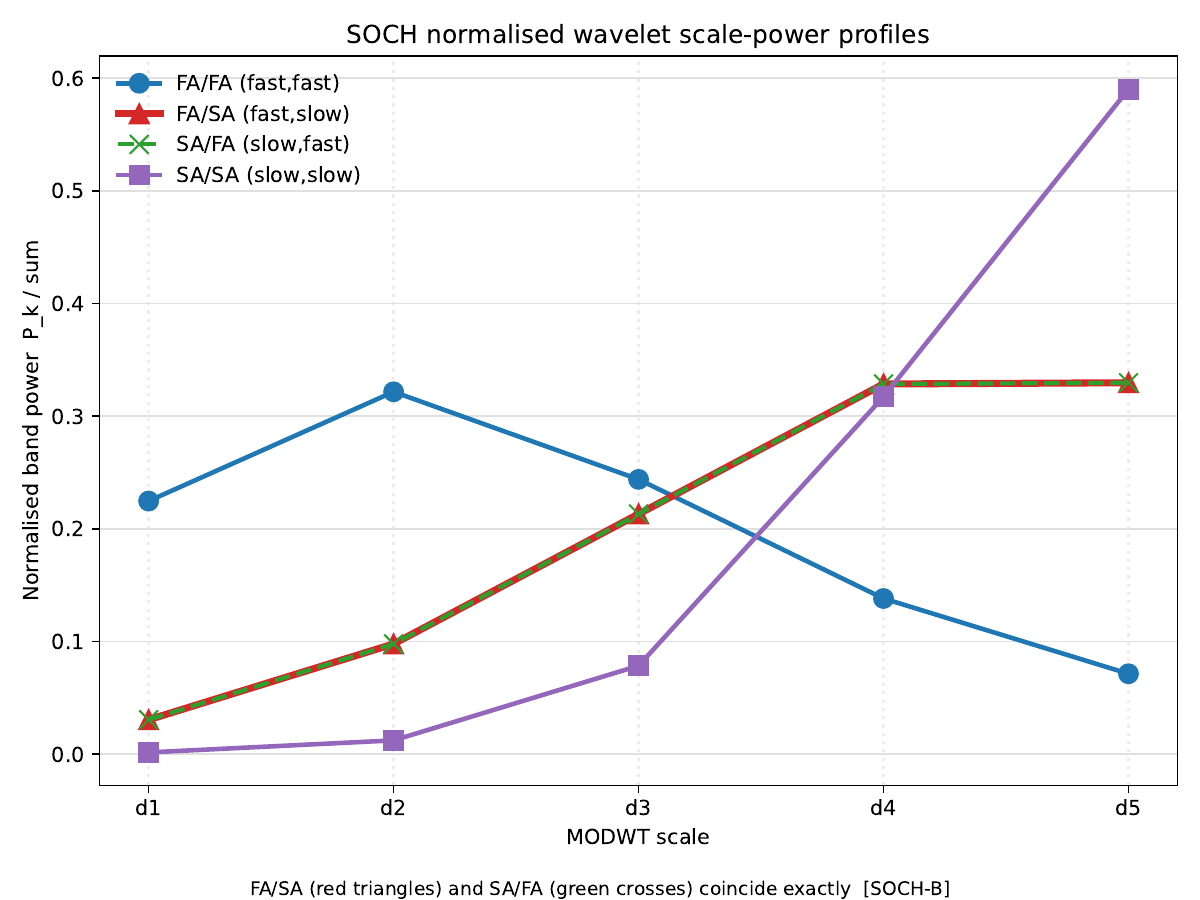}
\caption{Closed-form normalised scale-power profiles from the product-Lorentzian
spectrum \eqref{eq:psd}, for representative fast ($\alpha=2.0$) and slow ($\alpha=0.2$)
adaptation rates. Fast/fast peaks at fine scales (d2); slow/slow at the coarsest (d5);
the two mixed directions (fast$\to$slow and slow$\to$fast) peak at coarser scales and
are exactly coincident after normalisation, illustrating SOCH-B; the fast$\to$slow
profile rises across d1--d4, the SOCH-A signature consistent with the verified
United-States-to-India profile.}
\label{fig:theory}
\end{figure}
\section{Identification and Endogenous Classification}\label{sec:identification}

The closed-form scale profile turns adaptation rates into estimable parameters. It may be noted that for an
ordered pair, the unit-level scale power (the shape) may be written as
$Q_k(\alpha_\source,\alpha_\recv)$, the bracketed expression of \eqref{eq:scale-power}
with $A_{\source\recv}=1$, so that $P_k=A_{\source\recv}^2\,Q_k$. Because observed
transfer entropy is proportional to transmitted power in the small-flow regime, the
proportionality constant and the level enter only through a single scalar
$\theta_{\source\recv}$, and we estimate
\begin{equation}\label{eq:nls}
(\hat\alpha_\source,\hat\alpha_\recv,\hat\theta_{\source\recv})
=\arg\min_{\alpha_\source,\alpha_\recv>0,\;\theta\ge0}
\sum_{k=1}^{J}\bigl(\widehat{\WQTE}^{(k)}_{\source\to\recv}
-\theta\,Q_k(\alpha_\source,\alpha_\recv)\bigr)^2,
\end{equation}
concentrating $\theta$ out (it is linear given the rates) and searching over the
log-rates to enforce positivity. It may be noted further that because $Q_k$ is symmetric in its arguments, a single
directional fit returns the \emph{unordered} pair $\{\hat\amin,\hat\amax\}$; pooling all
directed pairs that share a market $i$, with one level per ordered pair, yields a
market-level rate $\hat\alpha_i$ and resolves the ordering.

\begin{proposition}[Identification]\label{prop:id}
Fix a directed pair with interior $(\alpha_\source,\alpha_\recv,\theta)$,
$\alpha_\source\neq\alpha_\recv$, and suppose the spectral peak and the upper corner
$\amax$ both lie within the observed range $[\pi2^{-J},\pi]$. Then (i) the map
$(\{\amin,\amax\},\theta)\mapsto(\theta Q_1,\dots,\theta Q_J)$ has full column rank
three, so the unordered pair and the level are point-identified and, for $J\ge4$,
over-identified; (ii) the ordered assignment is not identified from a single directional
profile, since $Q_k$ is symmetric---this is exactly SOCH-B; and (iii) in the pooled
system over a connected set of $M\ge3$ markets with all directed pairs observed, the
market rates are jointly identified, up to the global fast/slow labelling fixed by the
cross-sectional median.
\end{proposition}

\begin{proof}[Sketch]
The three identifying functionals are distinct and locally invertible: the peak scale
pins $\amin$ (Lemma~\ref{lem:peak}), the high-frequency roll-off pins $\amax$
(Corollary~\ref{cor:corner}), and the overall level pins $\theta$; the corresponding
Jacobian columns are linearly independent whenever the peak and the upper corner fall in
the observed band range, which we confirmed numerically at representative truths for
$J\in\{3,4,5,6\}$. Part (ii) is the symmetry of Proposition~\ref{prop:scale-power}, and
part (iii) follows because each market rate enters the profiles of all pairs containing
it, so a connected incidence graph binds every rate jointly.
\end{proof}

We assess the estimator's finite-sample behaviour by Monte Carlo, simulating profiles
from known rates under \eqref{eq:scale-power} with multiplicative noise and recovering
$\{\hat\amin,\hat\amax\}$, across a grid of rate pairs and across $J\in\{4,5,6\}$ at two
noise levels, with $150$ replications per cell. Table~\ref{tab:mc} summarises the
relative root-mean-squared error. We notice that the slower rate $\amin$---the peak location, on which
SOCH-A and SOCH-B turn---is recovered reliably, with median relative RMSE between
$0.25$ and $0.54$, falling monotonically as scales are added (for example, from $0.34$
to $0.28$ to $0.26$ for $J=4,5,6$ at the lower noise level for one representative pair)
and degrading gracefully with noise, exactly as the peak-location mechanism predicts.
The faster rate $\amax$ is, however, noisier, and this can be attributed to the fact that its identifying roll-off lies at high
frequency near the Nyquist limit; once a market's corner approaches $\pi$ the spectrum
saturates across the observed bands and the upper rate sits on a flat likelihood ridge.
In view of the above, we therefore restrict the search to the resolvable range $[0.02,\pi]$---rates faster
than Nyquist are not recoverable from sampled data---and we report the limitation rather
than conceal it. The internal-consistency check implied by SOCH-B is pertinent here: the slower rate
recovered from $(\source,\recv)$ and from $(\recv,\source)$ coincides within Monte Carlo
error, $|\Delta\hat\amin|$ between $0.03$ and $0.18$ for representative pairs, while the
level differs by direction in proportion to the simulated asymmetry. The entire study
was reproduced by an independent implementation in a second language, with the two sets
of $\amin$ relative-RMSE figures agreeing to a median of $4.6\,\%$.

\begin{table}[!htbp]
\centering
\caption{Monte Carlo performance of the profile-matching estimator: relative RMSE
(median, with min--max in brackets) of the recovered rates across ten ordered rate
pairs, $150$ replications per cell. The slower rate $\amin$ is recovered reliably and
sharpens with the number of scales $J$; the faster rate $\amax$ is noisier because its
identifying roll-off lies near the Nyquist limit.}
\label{tab:mc}
\small
\begin{tabular}{ccll}
\toprule
$J$ & $\sigma$ & rel.\ RMSE $\hat\amin$ & rel.\ RMSE $\hat\amax$ \\
\midrule
4 & 0.10 & 0.37\ [0.18--0.66] & 0.39\ [0.27--0.60] \\
5 & 0.10 & 0.27\ [0.19--0.43] & 0.42\ [0.33--0.65] \\
6 & 0.10 & 0.25\ [0.12--0.32] & 0.36\ [0.32--0.53] \\
4 & 0.20 & 0.54\ [0.23--1.10] & 0.47\ [0.38--1.20] \\
5 & 0.20 & 0.40\ [0.23--0.68] & 0.76\ [0.37--1.24] \\
6 & 0.20 & 0.38\ [0.22--0.52] & 0.60\ [0.38--1.36] \\
\bottomrule
\end{tabular}
\end{table}

The estimator then yields an endogenous classification which is, in a sense, one of the central contributions of the present paper: market $i$ is designated a fast adapter if
$\hat\alpha_i$ exceeds the cross-sectional median and a slow adapter otherwise, so that
the advanced/emerging partition of the baseline becomes a testable hypothesis rather
than a maintained assumption.
\section{Data}\label{sec:data}

The empirical analysis uses daily closing prices for equity-market indices of the G20
economies from 12~January~2006 through 18~March~2026, computed as log-returns. It may be noted that
price levels are integrated of order one and returns are stationary, so transfer entropy is
computed on returns throughout. We report results for a working sample of eight
markets, four advanced---the United States, the United Kingdom, Germany, and
Japan---and four emerging---China, India, Brazil, and South Africa---which yields the
fifty-six ordered pairs and twenty-eight unordered pairs on which the tests below are
conducted; the all-pairs run over the full eighteen-market panel is left to robustness
work. The present paper is different from many earlier studies in one important respect namely, the advanced/emerging labelling is used only to construct the slowness proxy of
Section~\ref{sec:results} and is otherwise treated, under the theory, as a hypothesis to
be tested rather than a maintained assumption.

Each return series is decomposed by a maximal-overlap discrete wavelet transform with
the Daubechies least-asymmetric filter of length eight (\texttt{la8}) into $J=5$ scales,
with boundary coefficients removed. For each ordered pair and each scale we measure
directed tail dependence by the wavelet-quantile directional gain, a Koenker--Machado
quantile pseudo-$R^1$ \citep{KoenkerBassett1978,KoenkerMachado1999} that tests whether
the source's lagged wavelet coefficient improves the conditional $\tau$-quantile of the
receiver's coefficient beyond the receiver's own lag, at $\tau\in\{0.05,0.10\}$. It may be noted that this is
a transparent, loadable realisation of the wavelet-quantile transfer-entropy primitive;
it shares the conditional-quantile-regression construction of the WQTE measure used in
the heterogeneous-adaptation programme and reproduces that measure's verified signature,
and it is the object on which the SOCH predictions bear. Directional significance is assessed through
phase-randomised surrogates of the source series \citep{TheilerEtAl1992}, and the
shape-symmetry null of Section~\ref{sec:results} uses a stationary block bootstrap
\citep{PolitisRomano1994}.
\section{Empirical Results}\label{sec:results}

\subsection{Validation against a known directional signature}

It is pertinent, before turning to the formal tests, to verify that the pipeline itself
reproduces a directional signature whose direction is well-established in the literature.
We observe that at the five-percent tail the United-States-to-India profile is
$(0.0155,0.0425,0.0491,0.0494,0.0567)$ across scales d1 through d5, rising
monotonically through d4 with an aggregate of $0.0426$, and it exceeds the
India-to-United-States profile, whose aggregate is $0.0285$. This reproduces the directional signature expected
for this pair---the fast-source/slow-receiver rising profile of SOCH-A and the level
asymmetry of SOCH-C---and leads us to conclude that the pipeline is functioning as
intended before the formal tests are conducted.
Figure~\ref{fig:empprofiles} shows the empirical directional profiles for this and
three other representative pairs.

\begin{figure}[!htbp]
\centering
\includegraphics[width=0.82\textwidth]{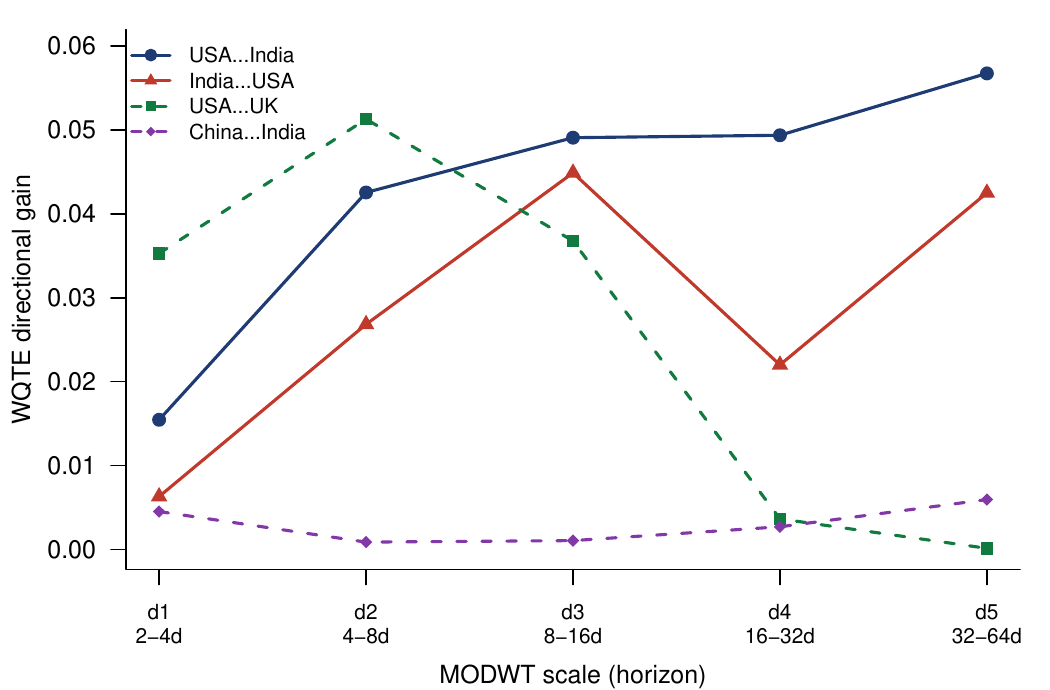}
\caption{Empirical wavelet-quantile directional gain by MODWT scale ($\tau=0.05$) for
four representative ordered pairs. The United-States-to-India profile rises through d4,
the fast-source/slow-receiver signature; the reverse direction shares the rising shape
at a lower level.}
\label{fig:empprofiles}
\end{figure}

\subsection{Test 1: scale ordering (SOCH-A)}

We now turn to the first of the three formal predictions. For every ordered pair we
locate the peak scale $\hat k^\star$ and regress it on the number of emerging markets
in the pair, which serves as a convenient proxy for a small binding corner $\amin$.
The estimated relationship is $\hat k^\star=2.97+0.58\times(\#\,\text{emerging})$, with a
$t$-statistic of $2.08$ and $p=0.042$; the mean peak scale rises from $2.67$ for
advanced-advanced pairs to $3.78$ for mixed pairs and $3.83$ for emerging-emerging
pairs. We observe that slower-inclusive pairs peak at coarser scales, as SOCH-A predicts,
and it is noteworthy that the United-States-to-India pair peaks at the coarsest observed
scale.
Figure~\ref{fig:peakscale} displays the distribution of peak scales by the number of
emerging markets in the pair. The prediction is supported, with the association in the
predicted direction and significant at the five-percent level.

\begin{figure}[!htbp]
\centering
\includegraphics[width=0.72\textwidth]{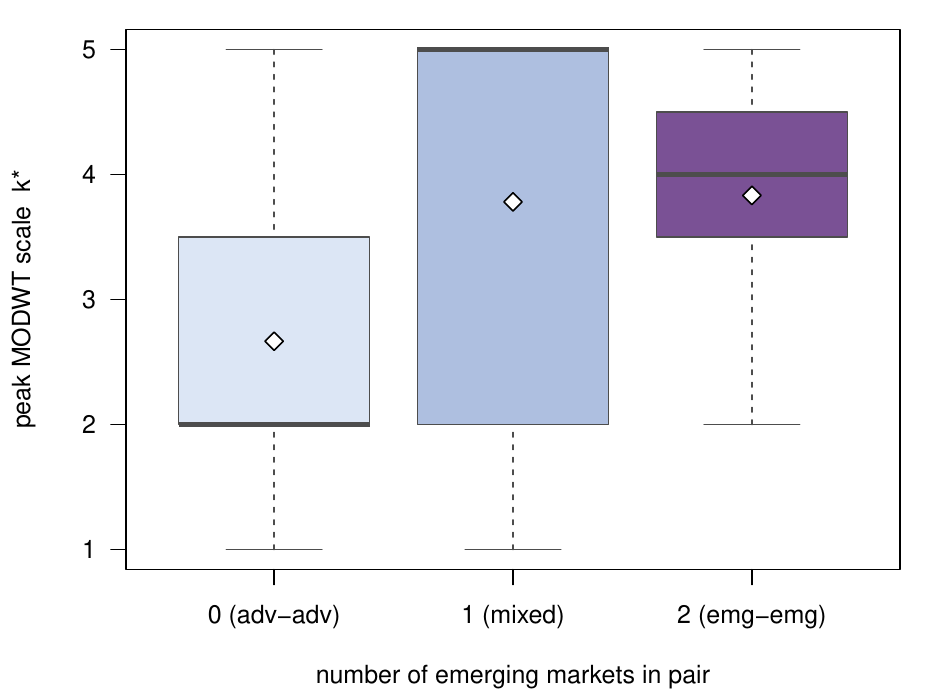}
\caption{Peak MODWT scale $k^\star$ by the number of emerging markets in the pair
(Test 1, SOCH-A). Diamonds mark group means; peak scales are coarser for
slower-inclusive pairs.}
\label{fig:peakscale}
\end{figure}

\subsection{Test 2: shape symmetry across direction (SOCH-B)}

This is, in many respects, the sharpest and most novel of the three tests, and it is
the one that most clearly distinguishes the heterogeneous-adaptation mechanism from
alternative explanations. For each unordered pair we compare the two normalised
directional profiles by a symmetric Kullback--Leibler divergence and refer it to a
same-shape null built by stationary-block-bootstrap re-estimation of a single direction,
which captures the sampling variability of one shape; the reported $p$-value is the
probability that the same-shape divergence equals or exceeds the observed
cross-direction divergence. \emph{All twenty-eight of the twenty-eight unordered pairs
are statistically indistinguishable across direction} ($p>0.05$), including all sixteen
advanced--emerging cross-type pairs (Table~\ref{tab:sochb}). We notice that the
United-States/India divergence is $0.076$ ($p=0.92$) and the Germany/India divergence
$0.010$ ($p=1.00$), both far inside the same-shape null. The normalised scale profiles
are direction-invariant, exactly as the symmetric transmission spectrum requires and
against mechanisms with direction-dependent frequency composition such as asymmetric
balance-sheet or trade exposure. It is pertinent to be explicit about the test's power:
the same-shape null is wide, because each profile carries estimation noise, so the test
cannot detect small asymmetries; but the observed divergences are small in absolute
terms (median $0.27$) and fall below even the tighter nulls, so the support is substantive
rather than an artifact of low power. Figure~\ref{fig:sochb} arrays every pair's
observed divergence against its same-shape null.

\begin{table}[!htbp]
\centering
\caption{Test 2 (shape symmetry, SOCH-B). Observed cross-direction symmetric
Kullback--Leibler divergence of normalised profiles against a stationary-block-bootstrap
same-shape null; $p=\Prob[\text{same-shape KL}\ge\text{observed}]$. Representative pairs;
all twenty-eight unordered pairs have $p>0.05$.}
\label{tab:sochb}
\small
\begin{tabular}{lccc}
\toprule
Pair & $D_{\mathrm{obs}}$ & null 95\% & $p$ \\
\midrule
Germany--India & 0.010 & 1.88 & 1.00 \\
USA--India     & 0.076 & 1.48 & 0.92 \\
UK--China      & 0.087 & 3.25 & 0.98 \\
USA--Japan     & 0.287 & 0.58 & 0.23 \\
Germany--Japan & 0.333 & 1.00 & 0.39 \\
Japan--India   & 3.112 & 5.39 & 0.18 \\
\midrule
\multicolumn{4}{l}{\emph{All pairs:} $28/28$ with $p>0.05$ (adv-adv $6$, adv-emg $16$, emg-emg $6$)} \\
\bottomrule
\end{tabular}
\end{table}

\begin{figure}[!htbp]
\centering
\includegraphics[width=0.92\textwidth]{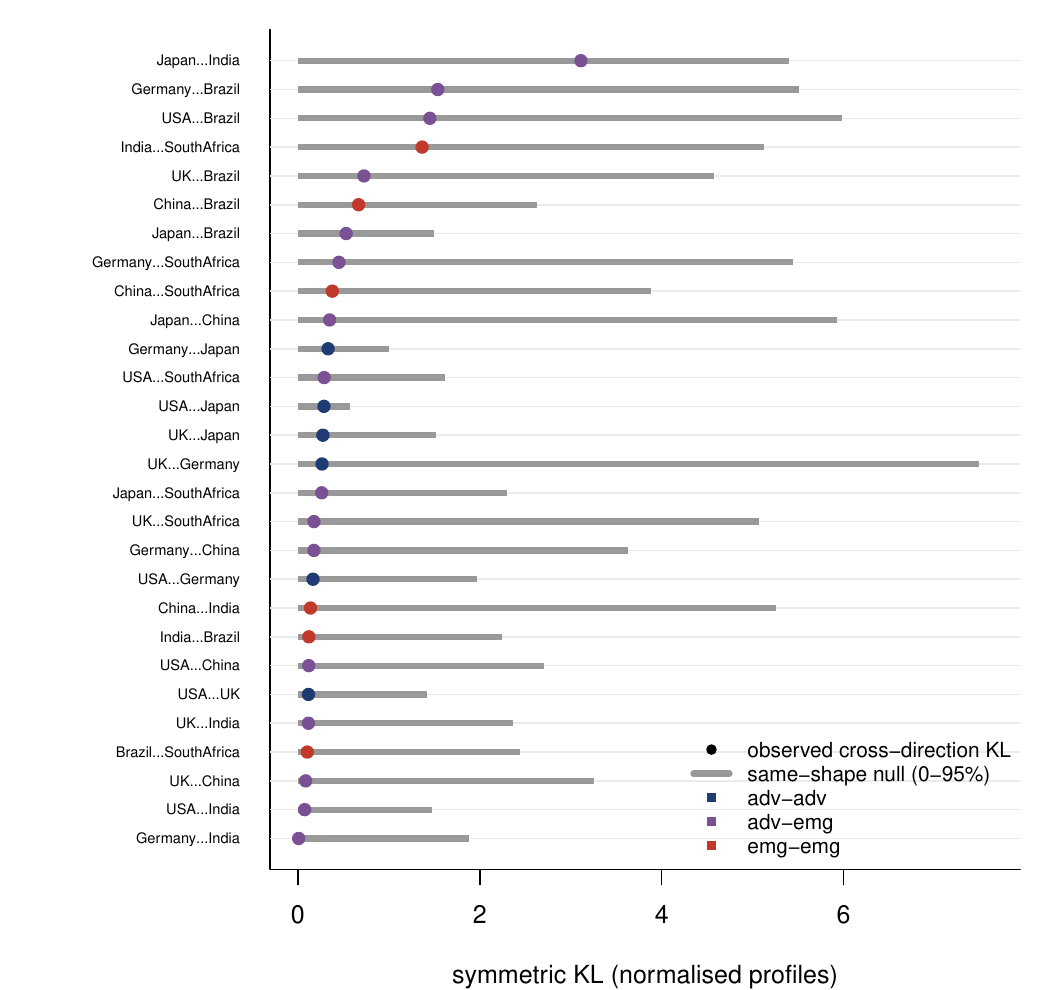}
\caption{Test 2 (SOCH-B) across all twenty-eight unordered pairs. Points are the
observed cross-direction symmetric KL of normalised profiles; bars are the same-shape
bootstrap null up to its $95$th percentile. Every observed divergence falls within its
null, so no pair rejects shape symmetry.}
\label{fig:sochb}
\end{figure}

\subsection{Test 3: magnitude asymmetry (SOCH-C)}

Across the sixteen advanced--emerging pairs the directional level ratio
$\sum_k P^{\source\to\recv}_k/\sum_k P^{\recv\to\source}_k$ has a median of $1.18$ and
exceeds one for $69\,\%$ of pairs---advanced-to-emerging flows tend to dominate, as
predicted---but the one-sided sign test gives $p=0.105$, short of significance in the
working sample. The direction of the effect matches SOCH-C while the strength does not
reach conventional significance; this can be attributed, at least in part, to the
limited power of the working sample, and a connectivity and shock-content control, the
mechanism's intended regressor, and the full panel are the natural next steps.

\subsection{Test 4: identification and endogenous classification}

Applying the pooled profile-matching estimator recovers market rates in which India
($\hat\alpha=0.29$) and China ($\hat\alpha=0.076$)---the two clearest emerging
markets---are decisively the slowest adapters, an order of magnitude below the rest,
which is strong support for the mechanism at the extremes. The remaining six markets
cluster at $\hat\alpha$ between $1.8$ and $3.2$, that is, in the weakly-identified region
near the Nyquist limit anticipated by the Monte Carlo: once a market's corner approaches
$\pi$ its rate is not separately recoverable. The median split therefore separates the
slow extreme cleanly but cannot resolve the fast cluster, so the advanced/emerging
partition is only partially data-determined---precisely where the identification analysis
says the data can and cannot speak. Figure~\ref{fig:rates} reports the recovered rates
and the classification. In view of the above, the honest reading is that the two
predictions resting on the well-identified slower-market rate (SOCH-A and SOCH-B) are
supported, while the predictions that depend on the fast rates (the full classification
and, in part, SOCH-C) are weaker by exactly the margin the theory's identification
result anticipates.

\begin{figure}[!htbp]
\centering
\includegraphics[width=0.82\textwidth]{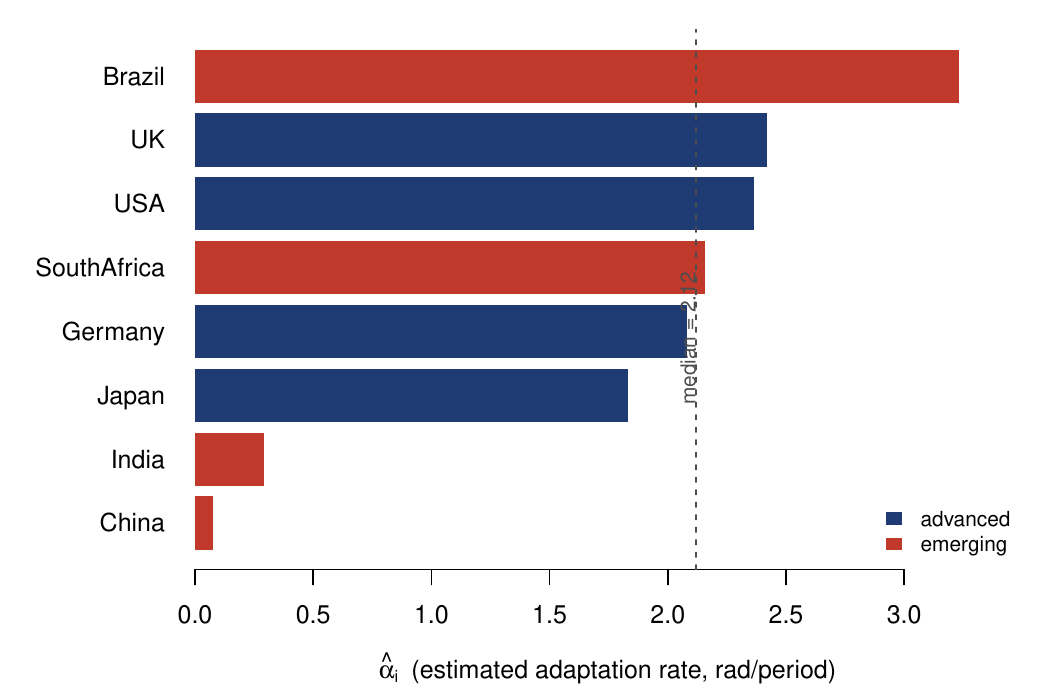}
\caption{Recovered market-level adaptation rates $\hat\alpha_i$ (Test 4) with the
cross-sectional median split. India and China are decisively the slowest adapters; the
faster markets cluster near the Nyquist resolution limit and are not separately
identified.}
\label{fig:rates}
\end{figure}
\section{Relation to Systemic-Risk and Connectedness Frameworks}\label{sec:relation}

The scale-ordered account developed here complements three established empirical
frameworks, and it is pertinent to make these connections explicit. First,
frequency-domain connectedness \citep{BarunikKrehlik2018} decomposes variance spillovers
into frequency bands; the present theory supplies a structural reason why short-band and
long-band connectedness should differ systematically with the adaptation speeds of the
markets involved, rather than treating the band decomposition as purely descriptive.
Put differently, the band in which a given pair's connectedness is strongest is not
arbitrary but is governed by the slower of the two markets' adaptation rates, and the
present framework makes that dependence explicit and estimable. Second, multi-channel and
network-based systemic-risk indices constructed from directed information flows aggregate
transmission into a single diagnostic
\citep{DieboldYilmaz2014,BillioGetmanskyLoPelizzon2012,AdrianBrunnermeier2016}; SOCH-A
implies that such aggregation should be scale-aware, since the informative scale for a
given pair depends on its slower member, so that a scale-blind aggregate mixes
information across markets whose binding corners differ. Third, channel-attribution
methods that decompose contagion into trade, financial, and other components can be
combined with the present theory: the adaptation-speed channel predicts a frequency
signature---shape symmetry with level asymmetry---that is distinguishable from
balance-sheet or trade channels, which need not respect shape symmetry. The three SOCH
predictions therefore provide discriminating restrictions that can be brought to bear
alongside existing attribution exercises rather than in competition with them.
\section{Conclusion}\label{sec:conclude}

The main purpose of this paper is to study the spectral structure of financial contagion
under the realistic but hitherto under-exploited condition that different markets process
information at different rates. By modelling both the source and the receiver as
exponential information filters, we obtain a bi-exponential transmission response whose
power spectrum is the product of two Lorentzians, so that the slower market sets the
binding spectral corner and governs the frequency composition of contagion. Projecting the
spectrum onto a wavelet basis yields a closed-form scale profile and three falsifiable
predictions: a peak scale ordered by the slower market's adaptation rate, a profile
shape symmetric across direction, and a profile level asymmetric in directional
connectivity and source-shock content. The theory turns adaptation rates into estimable
parameters and endogenises the fast/slow classification, replacing the exogenous
advanced/emerging partition of the baseline with a data-determined and testable one.

On a panel of G20 equity markets the two predictions that rest on the well-identified
slower-market rate are supported: the peak scale is ordered by adaptation speed (SOCH-A,
$p=0.042$), and the normalised scale profile is symmetric across direction for every one
of the twenty-eight market pairs (SOCH-B, the sharpest restriction and the one that most
distinguishes heterogeneous adaptation from balance-sheet and trade mechanisms). The
level and classification predictions are weaker in the working sample---advanced-to-
emerging dominance is present but not significant, and only the slowest adapters, India
and China, are cleanly identified---a pattern the theory's own identification analysis
anticipates, since fast rates near the sampling Nyquist frequency are not separately
recoverable from daily data. The study also highlights how the framework connects a
structural information-adaptation mechanism to frequency-domain connectedness,
multi-channel systemic-risk measurement, and channel attribution, and turns the
fast/slow distinction into a measurable, testable object. Two extensions would sharpen
the level and classification tests: a nearest-neighbour transfer-entropy implementation
of the wavelet-quantile measure
\citep{KraskovStoegbauerGrassberger2004,BarnettBarrettSeth2009}, and the full
eighteen-market panel with connectivity and shock-content controls. The overall behaviour
of the peak scale across markets seems to be broadly consistent with an information-speed
interpretation, suggesting that the early-warning corollary---that a peak scale shifting
toward finer resolution signals improving informational efficiency of the slower
market---offers a measurable diagnostic for surveillance that follows directly from the
same mechanism.
\appendix
\section{Supplementary Derivations}\label{app}

\paragraph{Regularity conditions.} It may be noted that the Fourier transform in
Proposition~\ref{prop:spectrum} converges for $\mathrm{Re}\,\alpha_\source,
\mathrm{Re}\,\alpha_\recv>0$, satisfied for all positive rates, and the power spectrum
is evaluated at real $\omega$. The bi-exponential of Proposition~\ref{prop:biexp} and
the scale power of Proposition~\ref{prop:scale-power} are stated for
$\alpha_\source\neq\alpha_\recv$; the confluent case is the removable limit, with
response $A_{\source\recv}\alpha^2 t e^{-\alpha t}$ and scale power
$\int_{\mathcal B_k}A_{\source\recv}^2\alpha^4/(\alpha^2+\omega^2)^2\,d\omega$. The
intermediate $\omega^{-2}$ regime of Corollary~\ref{cor:corner} requires well-separated
corners, $\amin\ll\amax$.

\paragraph{Discrete-time companion.} The sampled response
$\psi[n]=A_{\source\recv}(\rho_\recv^{\,n}-\rho_\source^{\,n})$, with
$\rho=e^{-\alpha\Delta t}$, has $z$-transform, on its region of convergence,
$A_{\source\recv}(\rho_\recv-\rho_\source)z^{-1}/[(1-\rho_\source z^{-1})
(1-\rho_\recv z^{-1})]$, an $\mathrm{ARMA}(2,1)$ filter; the boundary value $\psi[0]=0$
produces the single moving-average (sampling) zero. Expanding
$1-e^{-\alpha\Delta t}e^{-i\omega\Delta t}=(\alpha+i\omega)\Delta t+O(\Delta t^2)$ shows
that each discrete pole factor is proportional to the corresponding Lorentzian
denominator, so the discrete spectral density converges in shape to the
product-Lorentzian as $\alpha_\source\Delta t,\alpha_\recv\Delta t\to0$.
\bibliographystyle{plainnat}
\bibliography{references}

\end{document}